\documentclass[AMA,STIX2COL]{MRM}
\articletype{Article Type}%
\received{26 April 2016}
\revised{6 June 2016}
\accepted{6 June 2016}
\usepackage{amsmath,amssymb,amsfonts}
\usepackage{listings}
\usepackage{graphicx}
\usepackage{bbm}
\usepackage{soul}
\usepackage{amsmath}
\usepackage{amsthm}
\usepackage{textcomp}
\usepackage{color, colortbl}
\definecolor{Green}{rgb}{0.5,1,0.5}
\definecolor{Red}{rgb}{1,0.5,0.5}
\usepackage{mathtools}
\usepackage{textcomp}

\usepackage{rh_defs_21}

\usepackage{cancel} 
\usepackage{soul}

\newcommand\vtheta{{\bm \theta}}

\begin{document}
\title{\emph{k}-band: self-supervised MRI reconstruction trained only on limited-resolution data \\ 
}

\author[1]{Frederic Wang}{\href{0009-0000-1904-6242}{\includegraphics{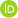}}}

\author[2]{Han Qi}{\href{0000-0002-3642-3596}{\includegraphics{figures/Orcidlogo.pdf}}}

\author[1]{Alfredo De Goyeneche}{\href{0000-0002-1550-280X}{\includegraphics{figures/Orcidlogo.pdf}}}

\author[3]{Reinhard Heckel}{\href{0000-0002-2874-2984}{\includegraphics{figures/Orcidlogo.pdf}}}

\author[1]{Michael Lustig}{}

\author[4]{Efrat Shimron}{\href{0000-0002-2267-0561}{\includegraphics{figures/Orcidlogo.pdf}}}

\authormark{Wang \textsc{et al}}

\address[1]{\orgdiv{Department of Electrical Engineering and Computer Sciences}, \orgname{University of California, Berkeley}, \orgaddress{\state{California}, \country{USA}}}

\address[2]{\orgdiv{Department of Computer Science}, \orgname{Harvard University}, \orgaddress{\state{Massachusetts}, \country{USA}}}

\address[3]{\orgdiv{Department of Computer Engineering}, \orgname{Technical University of Munich}, \orgaddress{\country{Germany}}}

\address[4]{\orgdiv{Department of Electrical and Computer Engineering and Department of Biomedical Engineering}, \orgname{Technion}, \orgaddress{\country{Israel}}}

\corres{Efrat Shimron, \email{efrat.s@technion.ac.il}}

\finfo{The authors would like to acknowledge funding from grants \fundingNumber{U24EB029240}, \fundingNumber{U01EB029427}, \fundingNumber{R01EB009690}, and \fundingAgency{GE healthcare}. E.S. acknowledges funding from the \fundingAgency{Weizmann Institute Women’s Postdoctoral Career Development Award in Science.}} 

\abstract[Summary]{
\section{Purpose} Although deep learning (DL) offers powerful tools for image reconstruction from undersampled MRI measurements, a major hurdle is the need for high-quality training data, which are difficult to acquire due to practical constraints. To address this, we introduce the \emph{k}-band framework, which enables training DL models using only partial, limited-resolution data, which can be acquired easily. Moreover, {k}-band offers test-time generalization to high-resolution reconstructions. 
\section{Methods} The \emph{k}-band framework designs the acquisition and training in-tandem. We propose the acquisition of \emph{k}-space bands, with limited resolution in the phase-encoding dimension; this is fast, practical, and easy-to-implement. To enable training using only these \emph{k}-space bands, we also propose an optimization method dubbed \emph{stochastic gradient descent over k-space subsets}. We prove analytically that this method stochastically approximates fully supervised training when the bands' angles are randomized across subjects and a \emph{k}-space loss-weighting mask is applied. 
\section{Results} Experiments with raw MRI data demonstrate that \emph{k}-band achieves performance on-par with that of fully supervised and self-supervised methods trained on high-resolution data, with the benefit of training on limited-resolution data. 
\section{Conclusion} The proposed \emph{k}-band framework offers practical strategies for fast acquisition and self-supervised training using limited-resolution data, with theoretical guarantees. It is easy-to-implement and agnostic to the pulse sequence and DL architecture. It can thus facilitate curation of new datasets and development of DL models for data-challenging regimes.
}

\keywords{Image reconstruction, self-supervised, unrolled networks, limited resolution. \vspace{-0.5cm}}


\maketitle

\section{Introduction}
\label{sec:introduction}

Although MRI offers superb image quality, its clinical application is restricted by its long acquisition time, which limits the spatio-temporal resolution and increases sensitivity to motion artifacts. Extensive research has focused on accelerating MRI, primarily through \emph{k}-space (Fourier-domain) undersampling and image reconstruction methods that infer the missing information. Well-established approaches include Parallel Imaging (PI) with multi-coil arrays \cite{pruessmann1999sense,Deshmane2015parallel,griswold2002generalized,lustig2010spirit,shimron2019core}, 
and Compressed Sensing (CS), which involves random undersampling and transform-domain regularization \cite{lustig2007sparse,vasanawala2011practical,feng2016xd,feng2017compressed,shimron2020temporal}. The latter can be learned using deep-learning (DL) frameworks, which have attracted substantial attention due to their ability to learn image priors in a data-driven manner \cite{lundervold2019overview,hammernik2022physics}. 

Nevertheless, a critical challenge that hinders the development of DL methods is the need for high-quality training data \cite{bell2023sharing}. Open-access databases of raw, high-quality MRI measurements are available but scarce, and offer only limited types of data \cite{knoll2020fastmri,ong2018mridata,calgary,desai2021skm}. Other databases offer non-raw MR images, but these are often preprocessed; training DL reconstruction models on them could lead to overly optimistic results and biased algorithmic evaluation \cite{shimron2022implicit}. Researchers thus often need to create new training databases from scratch, a difficult task due to the MRI slow acquisition rate. These challenges are particularly substantial in dynamic MRI (such as pulmonary \cite{jiang2018motion}, flow \cite{markl20124d}, and dynamic contrast-enhanced imaging \cite{zhang2015fast,ong2020extreme}), where it is often impractical to acquire full \emph{k}-space data and there is an inherent tradeoff between the spatial and temporal resolutions \cite{hammernik2022physics,nayak2022real}. 
 
Due to these challenges, there is a growing interest in self-supervised reconstruction, where networks are trained using only partial \emph{k}-space data \cite{akccakaya2022unsupervised,hammernik2022physics,chen2022ai}. 
Recent studies include methods based on the Noise2Noise\cite{lehtinen2018noise2noise} framework. One example is the 
 Self-Supervised learning via Data Undersampling (SSDU) \cite{yaman2020self_mrm} method, which has been analyzed analytically \cite{millard2023theoretical}. Other approaches include zero-shot learning  \cite{korkmaz2022unsupervised,yaman2021zero}  and training diffusion models on corrupted data\cite{daras2024ambient}.

Here we introduce the \emph{k}-band framework\footnote{This work is an extension of our conference abstracts \cite{sedona_kband} \cite{ismrm_kband}}, which aims to tackle current barriers by synergistically co-designing the data acquisition and training processes: (i) Acquisition: we propose to acquire \emph{k}-space bands, which have limited resolution in one phase-encoding (PE) dimension, and high resolution in the other dimensions. This approach is fast and easy to implement on standard MRI scanners. (ii) To further facilitate easier data acquisition, we propose to acquire only one \emph{k}-space band from each subject. Thus, training data can be collected rapidly, e.g. during dead time between clinical scans. (iii) Training: To enable self-supervised training using only the partial limited-resolution data, we introduce a novel training method dubbed \emph{stochastic gradient descent (SGD) over k-space subsets}. We derive this method analytically and show that it stochastically approximates fully supervised training (i.e., training with fully sampled high-resolution data) when two simple conditions are met: the bands' angles are randomized across subjects and a unique loss-weighting mask, derived analytically, is applied during training. (iv) Test-time generalization: we demonstrate, both analytically and via numerical simulations, that although training is done solely on limited-resolution data, during inference the network generalizes to high-resolution reconstructions.

\begin{figure*}[htp]
    \centering
    \includegraphics[width = \textwidth]{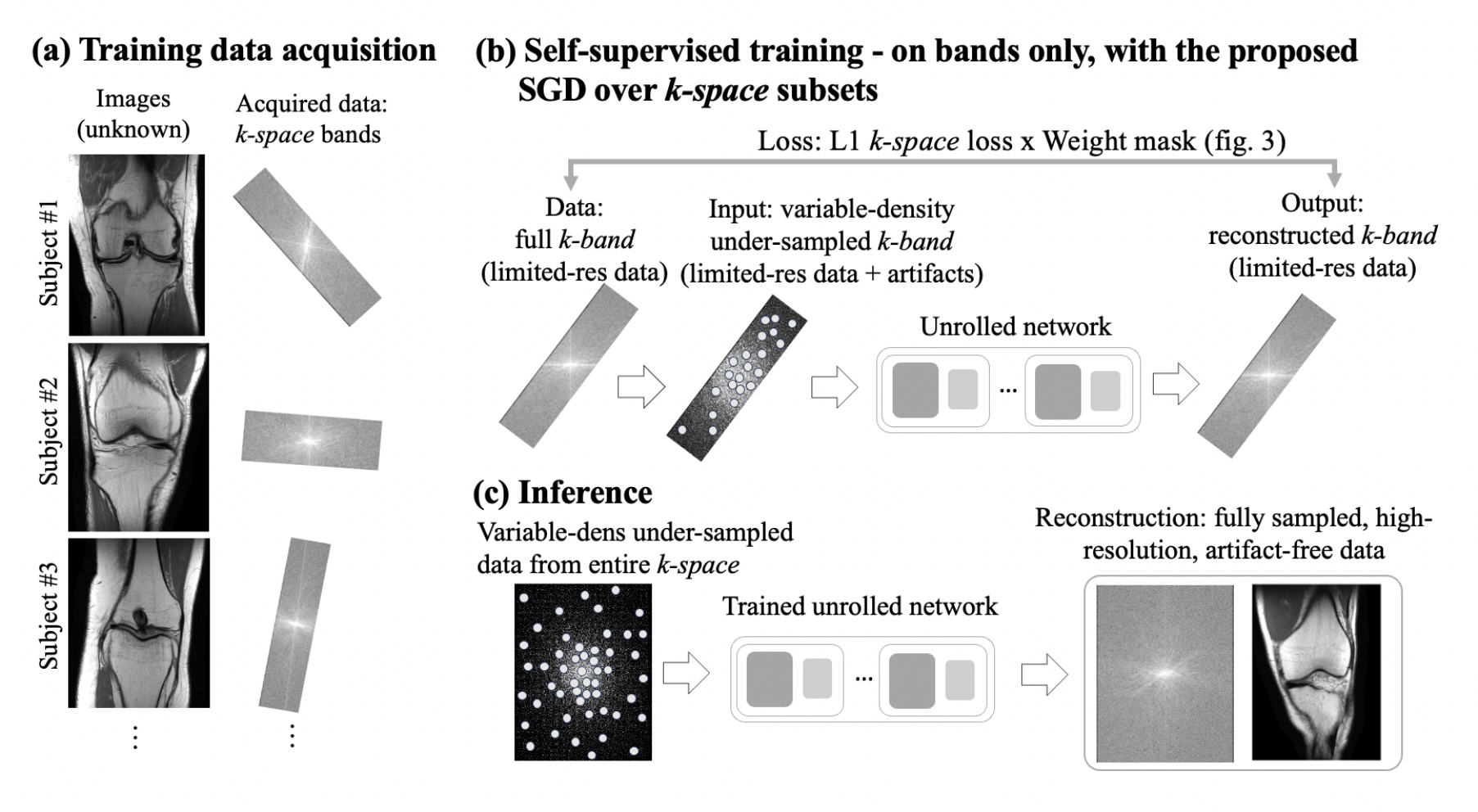}
    \caption{The proposed \emph{k}-band strategy. (a) Training data consists of \emph{k}-space bands, with different orientations across training examples. (b) During training, the band is undersampled with a variable-density mask and the network learns to reconstruct images in a self-supervised manner. (c) During inference, the network is not limited to limited-resolution data. It receives variable-density undersampled data from the entire \emph{k}-space, and reconstructs high-resolution images even though it never saw such examples during training.}
    \label{fig:flowchart}
\end{figure*}

\subsection{Background: The MRI inverse problem}

This work addresses image reconstruction from sub-Nyquist sampled \emph{k}-space measurements. The aim is to reconstruct an image $\vx$ from noisy, undersampled measurements $\vy=\mM \mF \vx + \ve$, where $\mF$ is an operator that describes the imaging system (the Fourier transform), $\mM$ is an operator that chooses the locations of the sampled \emph{k}-space pixels, and $\ve$ is noise. For multi-coil data, the $\mF$ operator will also include the coils' sensitivity maps.

A common approach to reconstruct the image from the measurements is to solve the regularized least-squares problem
\begin{align}
\label{eq:rls}
\underset{\vx}{\operatorname{argmin}} \|\mM \mF \vx - \vy\|_2^2 + R(\vx)
\end{align}

where $R(\vx)$ is a regularizer, such as a sparsity-promoting $\ell_1$ penalty~\cite{lustig2007sparse}. In DL frameworks, the reconstruction problem is commonly solved by training a neural network $f_\vtheta$ with parameters $\vtheta$ to map the measurements $\vy$ to a clean image. The network is typically trained in a supervised fashion, by minimizing a loss such as
\begin{align}
\label{eq:supervisedloss}
L(\vtheta) = \frac{1}{N} \sum_{i=1}^N \|\vx^{(i)} - f_\vtheta(\vy^{(i)}) \|_2^2
\end{align}

The training process hence requires pairs of measurements and target images.

The network architecture has substantial influence on the overall performance. At present, unrolled neural networks give strong performance \cite{hammernik2022physics}. This architecture can be obtained, for example, by unrolling the regularized least-squares objective (eq. ~\eqref{eq:rls}) with a first-order optimization algorithm:
\begin{align}
\vx_{k+1} = \vx_k - \alpha_k ( \frac{1}{2} \vy^T (\mM\mF \vx_k - \vy) + \nabla R(\vx_k))
\end{align}
Here, $\alpha_k$ is the step size (learning rate). The gradient of the regularizer (i.e., $\nabla R(\vx_k)$) is parameterized with a neural network (for example a U-net), which yields an unrolled neural network $f_\vtheta$.

\subsection{Overview of the \emph{k}-band framework}

\subsubsection{Data acquisition strategy}
\label{sec_data_acq_overview}

We propose to collect training data by acquiring \emph{k}-space bands, which have limited resolution in one direction. In 2D scans, this means acquiring data with high resolution along the readout dimension and limited resolution along the phase-encoding (PE) dimension (Figure \ref{fig:flowchart}a). We note that the proposed framework is also compatible with 3D imaging. For simplicity, here we describe the main concepts for a 2D acquisition, and the 3D setting will be discussed in section \ref{sec_acq_sampling}. 

The proposed bands acquisition scheme has many advantages. First, it offers speed: as only a portion of \emph{k}-space is sampled, bands can be acquired rapidly. This may allow for increasing the temporal resolution in dynamic scans or for faster volume coverage in high-dimensional scans. Secondly, bands can be acquired easily using any commercial MRI scanner, with various pulse sequences or image contrasts, by changing only two standard parameters: the PE resolution and the scan orientation. Therefore, this acquisition can be used to fill dead-time between scans in an MRI exam session and hence be easily integrated in a clinical workflow to obtain new datasets. Third, because the data within each band are fully sampled, this acquisition yields images that are \emph{free of global streaking or aliasing artifacts}. Such images may hence be more suitable for training DL models compared with images obtained from scattered sampling, e.g. variable-density (VD) sampling.

\subsubsection{Training strategy}

\emph{\textbf{SGD over \emph{k}-space subsets: key concept}}

The general, well-established SGD algorithm minimizes the fully supervised loss (eq. ~\eqref{eq:supervisedloss}). In each iteration one example $(\vx^{(i)}, \vy^{(i)})$ (or a mini-batch of examples) is chosen stochastically from the training set, and the gradient $\nabla L (\vx^{(i)}, f_{\theta} (\vy^{(i)}))$ is computed using \emph{fully sampled} \emph{k}-space data.

To enable training without full \emph{k}-space data, we propose to break down the SGD process into smaller steps. In each step we approximate the gradient using only a portion of the \emph{k}-space data. Notably, this framework is general and can be implemented with different acquisition schemes. Here, we explore one possible implementation involving bands acquisition, which has many benefits (section \ref{sec_data_acq_overview}).

\emph{\textbf{Training pipeline.}}

For each training example, we \emph{prospectively} acquire a band of \emph{k}-space data and then undersample it \emph{retrospectively} by applying a variable-density (VD) sampling mask. The network receives the VD-undersampled band as input, which has limited resolution \emph{and} undersampling artifacts. The training target is the fully sampled band (Figure. \ref{fig:flowchart}b), i.e. the entire \emph{k}-space data inside the band, hence the network only learns to remove the undersampling artifacts.

We compute the loss in \emph{k}-space, only inside the bands. This means that \emph{all} the data within the band is used for supervision, and there is \emph{no supervision} outside the band. We chose to avoid enforcing zeros outside the band, as this allows the network to learn priors for different \emph{k}-space regions. The network only "sees" one portion of \emph{k}-space for each training example. However, as band orientations change across examples, allowing for most of the \emph{k}-space to be covered across the training iterations, the network can generalize to high-resolution data.

\subsubsection{Inference}
Unlike the training data, which consist of bands, the \emph{test data} are acquired using a VD undersampling mask that covers the entire \emph{k}-space (Figure \ref{fig:flowchart}c). The VD statistics of the test data are the same as those of the VD undersampling applied during training. Inference is done using the pre-trained network without any re-training. This work aims to show that networks trained using only limited-resolution examples, i.e. \emph{k}-space bands, can generalize to high-resolution reconstructions at inference.


\subsection{Analytical derivation }

\subsubsection{Conditions for approximating supervised training}\label{proof1}

Here we derive the conditions under which \emph{k}-band stochastically approximates fully supervised SGD, i.e. approximates training with fully sampled \emph{k}-space data. Supervised training with a loss $L \colon \mathbb{C} \to \mathbb{R}$ applied in the Fourier domain minimizes an empirical version of the risk
\begin{align}
R(\vtheta) = \EX[(\vx, \vy)] {L(\mF f_\vtheta(\vy) - \mF \vx)}. 
\end{align}

Here, the expectation is over a joint distribution over the images $\vx$ and corresponding measurements $\vy$. Suppose the loss $L$ satisfies two conditions: (1) $L$ is pixelwise, i.e. $L(\vz) = \sum_j \ell([\vz]_j)$ where $[\vz]_j$ represents pixel $j$ of $\vz$. (2) If a pixel is zero, then the resulting loss is zero, i.e. $\ell(0) = 0$. Furthermore, let us define $\mU$ as a distribution of loss-supervision masks sampled independently of $(\vx,\vy)$, and a corresponding weighted loss as $L^W (\vz) = \sum_j \mW_{jj} \ell([\vz]_j)$ where $\mW$ is a deterministic loss-weighting mask. A stochastic gradient of this loss is
\begin{align}
G(\vtheta) 
= \nabla_\vtheta L^W (\mU (\mF f_\vtheta(\vy) - \mF \vx)).
\end{align}
Computing the gradient does not require knowledge of the image $\vx$, because an undersampled measurement $\mU\mF \vx$ is sufficient.

\begin{proposition}
Suppose that the matrix $\mW$ is a deterministic diagonal matrix obeying
\begin{align}
\mW_{jj} &= \frac{1}{\EX{\mU_{jj}}} 
\Longleftrightarrow \mW = (\EX{\mU})^{-1}.
\end{align}
Then 
\begin{align}
\EX[(\vx, \vy), \mU] {G(\vtheta)} = \nabla R(\vtheta).
\end{align}
\end{proposition}

The proposition states that \emph{SGD over k-space subsets} enables us to compute unbiased stochastic gradients of the risk. Thus, with sufficient training data, the network trained using our proposed method converges to one that is trained in a fully supervised fashion.

\begin{proof}
The expectation of the stochastic gradient is

\begin{align*}
\EX{G(\vtheta)}
&= \nabla_\vtheta \EX[(\vx,\vy),\mU]{L^W(\mU (\mF f_\vtheta(\vy) - \mF \vx)) } \\
&\stackrel{i}{=} \nabla_\vtheta \EX[(\vx,\vy),\mU]{\sum_{j} \mW_{jj} \ell(\mU_{jj} [\mF f_\vtheta(\vy) - \mF \vx]_{j})
} \\
&\stackrel{ii}{=} \nabla_\vtheta \EX[(\vx,\vy),\mU]{\sum_{j} \mW_{jj} \mU_{jj} \ell([\mF f_\vtheta(\vy) - \mF \vx]_{j})
} \\
&\stackrel{iii}{=}  \nabla_\vtheta \EX[(\vx,\vy)]{
\sum_j \mW_{jj} \EX[\mU]{\mU_{jj}}
\ell([ \mF f_\vtheta(\vy) - \mF \vx]_j ) 
} \\
&\stackrel{iv}{=} \nabla_\vtheta \EX[(\vx,\vy)]{
\sum_j 
\ell([ \mF f_\vtheta(\vy) - \mF \vx]_j ) 
} \\
&= \nabla_\vtheta R(\vtheta),
\end{align*}
as desired. $(i)$ uses the definition of the loss $L$, $(ii)$ uses the fact that $U_{jj} \in \{0, 1\}$ and $\ell(0) = 0$, $(iii)$ uses the independence of $\mU$ and $(\vx, \vy)$, and $(iv)$ uses the assumption on the masks. 

\end{proof}


\begin{corollary} \label{W_proof} Let $\mU = \{\mB_i\}$ where $\mB_i$ represents a band-acquisition mask that has orientation $i \in \{1, ..., k\}$. \emph{K}-band stochastically approximates fully supervised training with element-wise \emph{k}-space loss weighting when
\begin{align}
\label{eq:w_mask}
\mW = k (\sum_{i=1}^{k} \mB_i)^{-1}
\end{align}
\end{corollary}

\begin{proof}
We can express the \emph{k}-band training gradient using the stochastic gradient formulation. Suppose we acquire \emph{k}-space with band mask $\mB_i$, and suppose $\mV$ is an arbitrary VD undersampling mask. Then, the \emph{k}-band training loss can be written as 
\begin{align}
L^W (\mB_i (\mF f_\vtheta(\mB_i \mV \mF \vx) - \mF \vx))
\end{align}

By Proposition II.1, we can approximate fully supervised training by weighting the loss element-wise in the Fourier domain with 
\begin{align}
\mW = (\EX{\mU})^{-1} = (\EX{\mB_i})^{-1} = k (\sum_{i=1}^{k} \mB_i)^{-1}
\end{align}
where the last step is because each band mask $\mB_i$ is sampled uniformly from all possible band orientations.
\end{proof}


\subsubsection{Analysis of $\ell_1$ and $\ell_2$ loss functions}

\begin{corollary}
Both $\ell_1$ and $\ell_2$ loss functions (in \emph{k}-space) satisfy the conditions for $L$. Using our definition of $L(\vz) = \sum_j \ell(\vz_j)$, an $\ell_1$ loss can be defined with $\ell(x) = |x|$ which satisfies $\ell(0) = 0$. Additionally, an $\ell_2$ loss can be defined with $\ell(x) = \|x\|_2^2$, which satisfies $\ell(0) = 0$.
\end{corollary}

{\bf Remarks:}
Suppose we use the $\ell_2$ loss, i.e.,  $L(\vz) = \norm[2]{\vz}^2$. Then by Parseval's Theorem, we have that training in the Fourier domain is equivalent to training in the image domain, as follows:
\begin{align*}
R(\vtheta) 
= \EX{\norm[2]{\mF f_\vtheta(\vy) - \mF \vx}^2 }
= \EX{\norm[2]{ f_\vtheta(\vy) - \vx}^2 }.
\end{align*}

\begin{proposition}
Suppose that $\|\nabla_\theta f_\vtheta(\vy)\|^2_2 \leq M$, i.e. the network is $M-$Lipschitz, and suppose that $| \nabla_x l(x) | \leq D$, i.e. the loss function is $D-$Lipschitz. Let $C = \max_j (\mW_{jj})$. We have $$\EX [\vx, \vy, \mU] {\|G(\theta)\|^2} \leq  C D^2 M$$

It is important to choose stochastic gradients that have a small variance, as they result in stochastic gradients converging faster and thus requiring fewer training examples\cite{Klug_Atik_Heckel_2023}. This proposition allows us to quantify the variance of these stochastic gradients as a function of the choice of loss. Note that $C$ and $M$ are  agnostic to the loss function. 

\end{proposition}
\begin{proof}
Let $\vz = \mF f_\vtheta(\vy) - \mF \vx$ for simplicity. We have
\begin{align*}
&\EX [(\vx, \vy), \mU] {\|G(\theta)\|^2} = \EX [(\vx, \vy), \mU] {\|\nabla_\theta L^W (\mU \vz)\|^2} \\ 
&= \EX [(\vx, \vy), \mU] {( \sum_i \mW_{ii} \mU_{ii} \nabla_\theta \ell ([ \vz ]_j) )^2 } \\
&= \EX [(\vx, \vy), \mU] {\sum_i \sum_j \mW_{ii} \mU_{ii} \mW_{jj} \mU_{jj} \nabla_\theta \ell ([ \vz ]_i) \nabla_\theta \ell ([ \vz ]_j) } \\
&\stackrel{i}{=} \EX [(\vx, \vy), \mU] {\sum_i \mW_{ii}^2 \mU_{ii}^2 (\nabla_\theta \ell ([ \vz ]_i))^2 } \\
&\stackrel{ii}{=} \EX [(\vx, \vy)] {\sum_i \mW_{ii} (\nabla_\theta \ell ([ \vz ]_i))^2 }
\\ &= \EX [(\vx, \vy)] { \sum_i \mW_{ii} (\nabla_\theta \ell ([\mF f_\vtheta(\vy) - \mF \vx]_i))^2} \\
&\stackrel{iii}{=} \EX [(\vx, \vy)] { \sum_i \mW_{ii} ( \ell'([\mF f_\vtheta(\vy) - \mF \vx]_i) [\mF \nabla_\theta f_\vtheta(\vy)]_i )^2 } \\
&\stackrel{iv}{\leq} D^2 \EX [(\vx, \vy)] { \sum_i \mW_{ii} [\mF \nabla_\theta f_\vtheta(\vy)]_i ^2 } \\
&\stackrel{v}{\leq} C D^2 \EX [(\vx, \vy)] { \|\mF \nabla_\theta  f_\vtheta(\vy) \|_2^2 } \\
&\stackrel{vi}{=} C D^2 \EX [(\vx, \vy)] { \|\nabla_\theta  f_\vtheta(\vy) \|_2^2 } \\
&\stackrel{vii}{\leq} C D^2 M
\end{align*}
as desired. Equation $(i)$ comes from the fact that $\EX{\mW_{ii} \mU_{ii}} =\EX{\mW_{jj} \mU_{jj}} = 1$ and that those terms are independent for $i \neq j$, leaving only the diagonal squared terms after simplification. Equation $(ii)$ uses the fact that $\mU_{ii} \in \{0, 1\}$, so $\EX{\mU_{ii}^2} = 0^2 \mathbb{P} (\mU_{ii} = 0) + 1^2 \mathbb{P} (\mU_{ii} = 1) = \EX{\mU_{ii}} = \mW_{ii}^{-1} = \EX{\mW_{ii}^{-1}}$. Equation $(iii)$ uses the chain rule, equation $(iv)$ uses our assumption on the Lipschitz condition on the loss, equation $(v)$ comes from our assumption on the maximum value of $\mW$, equation $(vi)$ comes from Parseval's Theorem, and equation $(vii)$ comes from our assumption of the Lipschitz condition of the neural network.
\end{proof}


\begin{corollary} \label{L1better} By using a loss function with a bounded gradient, there will be smaller variance in the gradients. Under this regime, an $\ell_1$ \emph{k}-space loss ($D = |\nabla_x \ell_1(x)| = 1$) is expected to perform better than an $\ell_2$ \emph{k}-space loss ($D = |\nabla_x \ell_2(x)| = |x|$).
\end{corollary}

\subsubsection{Summary}

We showed that the proposed method stochastically approximates training with SGD on fully sampled high-resolution data when two simple conditions are met: (i) the limited-resolution axis is chosen randomly and uniformly for every new scan, hence \emph{k}-space is fully covered across the entire training set; (ii) The \emph{k}-space loss function is multiplied by the deterministic weighting mask $\mW$ (eq. (9)), which facilitates accurate learning of high-resolution details. We also found that training \emph{k}-band using an $\ell_1$ \emph{k}-space loss better approximates fully supervised training than training with an $\ell_2$ \emph{k}-space loss.

\section{Material and methods}

\subsection{Training: optimization and loss} 

\textbf{\emph{Unrolled network}}. We demonstrate the framework with an unrolled network, similarly to MoDL \cite{aggarwal2018modl}. However, we use self-supervised training on bands and a weighted \emph{k}-space loss.

Our unrolled optimization scheme is formulated by
\begin{align}
\vx_{n+1} &= (\mF^H \mB^H \mB \mF + \eta \mI)^{-1} ((\mB \mF)^H \vy + \eta \vz_n) \\
\vz_n &= \text{ResNet}_\theta (\vx_n)
\end{align}
where $H$ is the Hermitian transpose, $\eta$ is a regularization parameter,  and $\mB \in \{\mB_i\}_{i=1,...,k}$ is a band sampling operator (see Corollary \ref{W_proof}). $\text{ResNet}_\theta$ describes the network contained within one block (one iterate) of the unrolled network $f_\theta$. The network is trained end-to-end, and the loss is computed using the output of the last network layer.

\textbf{\emph{Loss function}}. We use a weighted $\ell_1$ \emph{k}-space loss, as it performs better than an $\ell_2$ \emph{k}-space loss both in theory (Corollary \ref{L1better}) and empirically (Table \ref{table:table1_loss}). The weighted $\ell_1$ \emph{k}-space loss is formulated by, 
\begin{align}
L(\vtheta) = \frac{1}{N} \sum_{i=1}^{N} \|\mW \mB_i (\mF f_{\theta}(\vy^{(i)}
) - \mF \vx^{(i)} ) \|_1
\end{align}
where $\vx^{(i)}$ are the band-limited \emph{k}-space data of sample $i$, $\vy^{(i)}$ are the corresponding VD-undersampling input, $\mW$ is the loss weighting derived above (see eq. \eqref{eq:w_mask}) and $\mB_i$ is the binary band mask.

\textbf{\emph{Loss-weighting mask}}. 
The loss-weighting mask $\mW$ has a deterministic formulation (eq. (8), which depends on two acquisition parameters: the band width and the number of distinct bands $k$.

For illustration, we display $\mW$ masks (Figure \ref{fig:W_mask}) for a band area that covers 12.5\% of \emph{k}-space and (a) an extreme case, where only 10 band orientations were acquired, (b) a case where 30 orientations were acquired, and (c) a case where 180 orientations were acquired. Figure \ref{fig:W_mask}d shows a profile of case (c). For all scenarios, $\mW$ obtains low values in the center of \emph{k}-space, and high values in its periphery. Intuitively, this loss-weighing compensates for over-exposure of the network to low-frequency
\emph{k}-space data and enhances learning in the peripheral areas.

\begin{figure}[h]
   \centering
   \includegraphics[width = 0.45\textwidth]{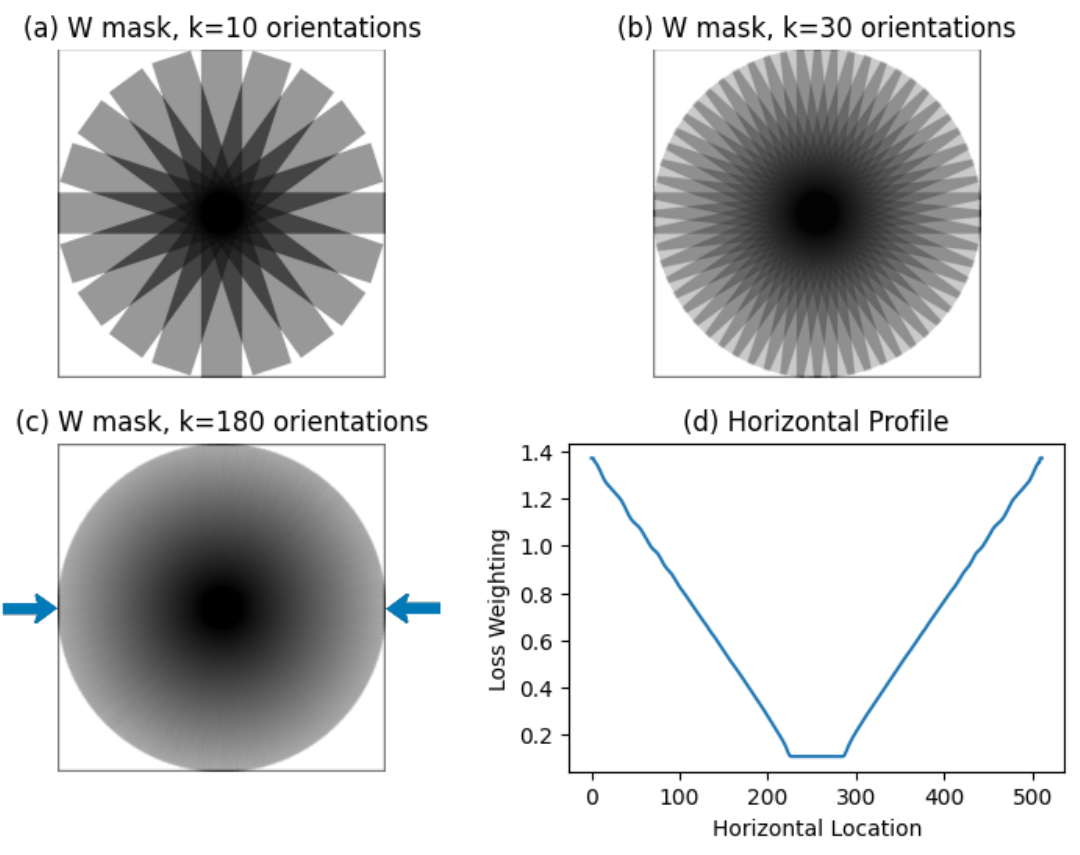}
   \caption{Examples of loss-weighting masks ($\mW$), where the number of distinct bands was (a) $k=10$, (b) $k=30$, (c) $k=180$. (d) Horizontal profile of the $\mW$ mask shown in (c), along the line between the two arrows. Notice that the loss-weighting mask inhibits the loss values in the \emph{k}-space center and enhances them in the \emph{k}-space periphery. Therefore, this mask facilitates equal learning across the entire \emph{k}-space.} 
   \label{fig:W_mask}
\end{figure}


\subsection{Acquisition and sampling}
\label{sec_acq_sampling}

We emphasize again that the training and test data are acquired differently. For training, whole bands are acquired \emph{prospectively}, and the VD sampling inside the bands area is applied \emph{retrospectively}, to train the network. For inference the VD undersampling is applied \emph{prospectively}. The \emph{k}-band framework is compatible with both 2D and 3D MRI scans (Figure \ref{fig3:band_mask}).

\emph{\textbf{2D MRI.}} The 2D setting is illustrated in the middle row of Figure \ref{fig3:band_mask}. For \emph{training}, we acquire bands by limiting the resolution along the PE dimension, and keeping the readout fully sampled. Then, we apply VD undersampling retrospectively along the PE dimension only; this is referred to as \emph{1D undersampling}.  For \emph{inference}, data are acquired with VD undersampling along the PE dimension, and the VD mask covers the entire k-space.

\emph{\textbf{3D MRI.}} The 3D case is illustrated in the bottom row of Figure \ref{fig3:band_mask}.  For \emph{training}, bands are acquired by limiting the resolution only along one PE dimension, while the readout and the second PE dimension are fully sampled. For \emph{inference}, however, data are acquired by undersampling both of the PE dimensions. This results in a 2D VD mask, hence we refer to this as \emph{2D undersampling}.

\begin{figure}[h]
    \centering
    \includegraphics[width = 0.49\textwidth]{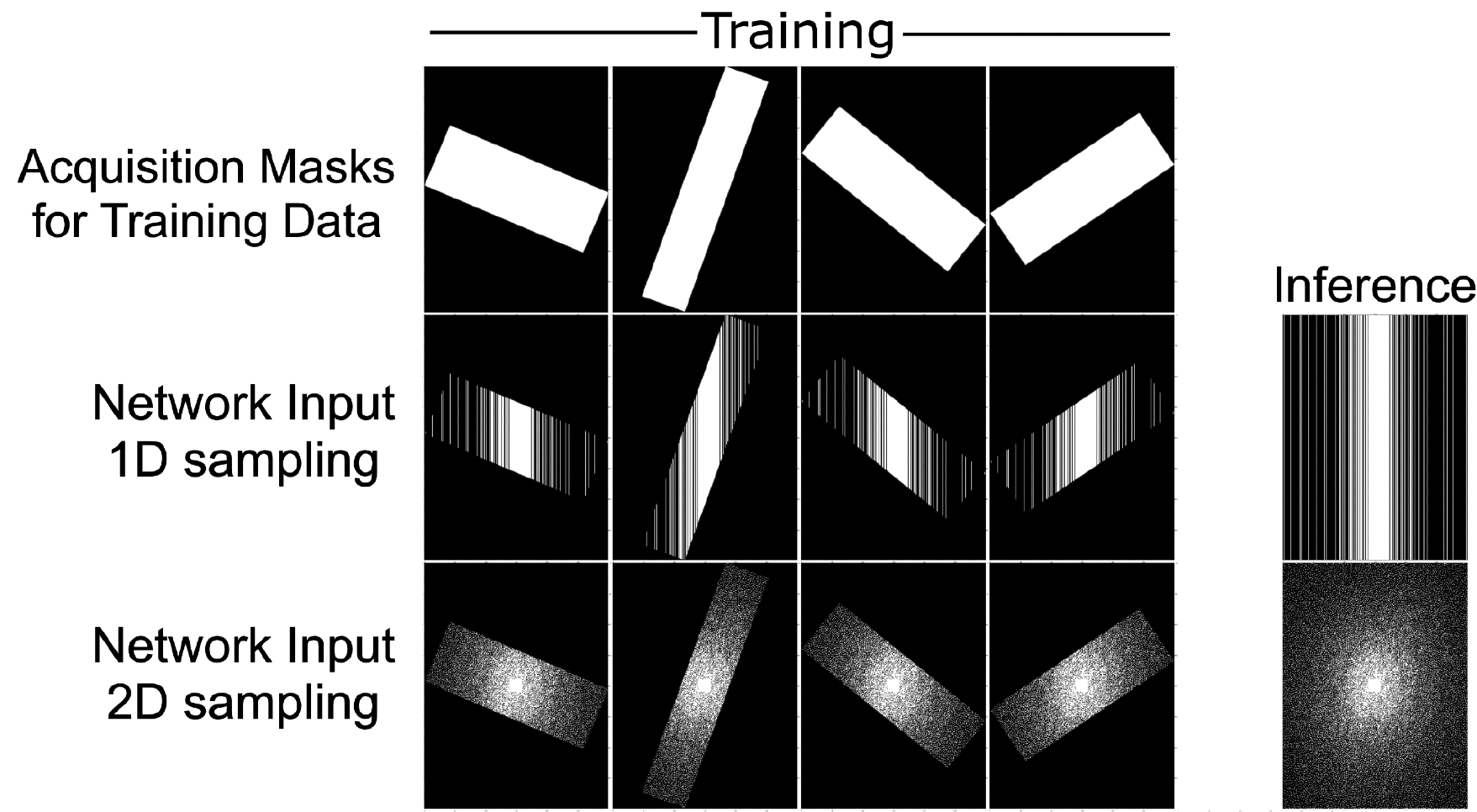}
    \caption{(Top) Visualization of prospective sampling masks for training data acquisition and loss supervision. (Middle) Retrospective 1D Poisson disc undersampling mask used for input to the network during training and inference. (Bottom) Retrospective 2D Poisson disc undersampling mask used for input to the network during training and inference.}
    \label{fig3:band_mask}
\end{figure}

\emph{\textbf{Terminology.}} We use $R_{band}$ to designate the scan time acceleration achieved by acquiring \emph{k}-space bands rather than the full \emph{k}-space. In other words, $R_{band}$ is the ratio of the full \emph{k}-space size over the band size. We use $R_{vd}$ to designate the acceleration obtained by the VD undersampling. It is worth noting that during training the VD sampling is retrospectively applied only to the bands data, and during inference it is applied to the entire \emph{k}-space.  

\subsection{Implementation Details}

Here we describe the implementation details for our experiments, including the data processing pipeline, sampling methods, and network architecture.

\subsubsection{Data preparation}

To avoid conducting any data crimes \cite{shimron2022implicit}, we provide a full description of the data preparation pipeline. We obtained raw multi-coil knee and brain data from fastMRI \cite{knoll2020fastmri}. Specifically, we used proton-density fat-saturated knee data, and T1-weighted, T2-weighted and FLAIR brain data. Each scan contains 20-30 slices. We used only the 10 middle slices from each scan (where structure is abundant), and normalized the intensity of each slice  to its 95th percentile. Each dataset included 1600 training examples and 400 test images.

When trying to train a network with several unrolls on our GPUs, we ran into computational issues related to the large size of the fastMRI raw data, e.g. 640x372x15 (height x width x coils) pixels. To mitigate these issues, we reduced the data size by applying coil-combination and external image-domain cropping. However, we emphasize that our method is not limited to coil-combined or single-coil data.

We computed the coil sensitivity maps using BART's implementation of the ESPIRiT algorithm \cite{uecker2014espirit,uecker2015berkeley}, with a calibration region of 20x20 (21x21) for the knee (brain) data. Then, single-coil datasets were generated using a standard coil combination method: $x = \sum_{i=1}^{n_{coils}} S_{i}^{*} X_i$ where $S_{i}$ is the sensitivity map of coil $i$. Next, the coil-combined images were cropped in the image domain to 400x300 (320x230) pixels for the knee (brain) data. This cropping mainly removed background areas. Only the training examples were cropped; we emphasize that the inference data were not cropped, to simulate realistic clinical scans. We used fully convolutional networks, hence inference can be done on images of any size. 

\subsubsection{Sampling implementation}

\textbf{\emph{Bands acquisition}}.
A new band mask was chosen uniformly at random for every training example, to simulate realistic acquisition (Figure \ref{fig3:band_mask}, top row). These masks are applied to the \emph{k}-space data, containing ones inside the chosen band and zeros outside it. We generated those masks only once, before training, and saved them together with the \emph{k}-space data to prevent sampling multiple bands from the same slice in different training iterations.

In the "vanilla" implementation of \emph{k}-band, the band orientation angle was uniformly drawn from integer values in the range of [0,180] degrees. We also conducted experiments showing that the framework can work well with less band orientations. In those cases, we defined $k$  as the number of possible (distinct) band orientations, and uniformly sampled band angles from this set: $[0, \frac{180}{k}, 2\times\frac{180}{k}, \cdots, (k-1)\times\frac{180}{k}]$.

To keep $R_{band}$ constant across all training examples, regardless of the randomly chosen orientation, the band width was varied with the angle, i.e. bands that were closer to a horizontal orientation were slightly wider than bands closer to a vertical orientation.

\textbf{\emph{VD sampling}}. 
The VD masks were generated using SigPy \cite{ong2019sigpy}. To maintain the VD statistics across all the examples in the training and test sets, the VD masks were generated independently of the binary band masks used for training.

\subsubsection{Network implementation}

For the unrolled network, we used the MoDL architecture \cite{aggarwal2018modl}, as implemented in the DeepInPy toolbox \cite{tamir2020deepinpy}. This architecture includes a ResNet consisting of 64 hidden channels and 3x3 kernels. Following the steps described above for data size reduction, our networks included 7 unrolls and 6 ResNet blocks for the brain data, and 10 unrolls and 8 ResNet blocks for the knee data. We tuned the learning rate using a grid search in the range of $1e^{-6}$ to $1e^{-3}$; a learning rate of $1e^{-4}$ was found to be best for all the  methods described in this paper. All of our experiments were done with TITAN Xp, Titan X (Pascal), and NVIDIA GeForce RTX 3090 GPUs. 

\subsection{Implementation details of other methods}

We compared \emph{k}-band with four other methods. The first two are trained using high-resolution data. Those are MoDL \cite{aggarwal2018modl}, which is fully supervised, i.e. trained with fully sampled \emph{k}-space data, and SSDU \cite{yaman2020self_mrm} which is a self-supervised method, trained with undersampled but high-resolution \emph{k}-space data. Note that \emph{k}-space data can be undersampled but still have high resolution, as in SSDU, or be fully sampled but have limited resolution, as in \emph{k}-band. We used the SSDU public GitHub code \cite{yaman2020self_mrm}, and tuned the hyperparameters accordingly. Specifically, the optimal ratio of \emph{k}-space samples used for data consistency vs. loss calculation was found to be $\rho=0.4$, as suggested by the authors. 

We also examined two self-supervised methods trained using only limited-resolution data. One of them, which we term \emph{k}-square, is trained using only low-resolution data acquired in a square around \emph{k}-space center (Figure \ref{fig:fig4_AllMasks}). This method was proposed in the SSDU paper as a benchmark. The other method, dubbed \emph{k}-vertical, 
is trained using data from a \emph{k}-space band that has a fixed vertical orientation in \emph{k}-space (Figure \ref{fig:fig4_AllMasks}). The \emph{k}-square and \emph{k}-vertical methods were both trained in a self-supervised manner, similar to \emph{k}-band: the input to those networks was data from the square/vertical band, undersampled retrospectively using a VD mask, and supervision was done using all the data in the band.

We have taken careful steps to ensure a fair comparison of the above methods. First, all the self-supervised methods examined here (\emph{k}-band, SSDU, \emph{k}-square, and \emph{k}-vertical) were trained using the same amount of data (MoDL is different because it uses full \emph{k}-space data). Secondly, during inference all five methods (MoDL and the four self-supervised methods) received the same type of data, where the VD mask covered the entire \emph{k}-space area (Fig \ref{fig3:band_mask}, middle-right and bottom-right). These methods thus differed only in their training strategies. Third, the five methods were all implemented with the unrolled network architecture described above. Moreover, the hyper-parameters were tuned for each method separately. The methods were evaluated by computing the Normalized Mean Square Error (NMSE) and the Structural Similarity Index Measure (SSIM) \cite{wang2004image} using the magnitude of the reconstructed and ground-truth images.  

\begin{figure}[h]
    \centering
    \includegraphics[width = 0.49\textwidth]{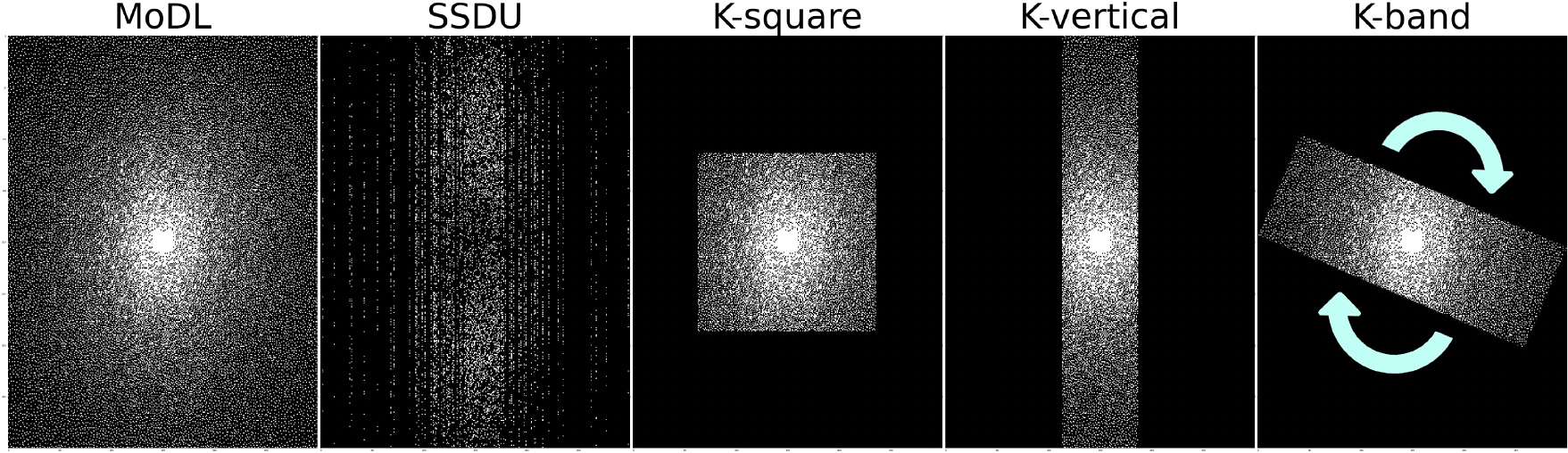}
    \caption{Sampling masks used for input, for the five methods discussed in this work. In these examples $R_{band}=R_{vd}=4$.}
    \label{fig:fig4_AllMasks}
\end{figure}

We investigated different loss functions for both \emph{k}-band and MoDL. For \emph{k}-band, we compared $\ell_1$ and $\ell_2$ \emph{k}-space loss functions, both implemented with the weighting described above, and observed that $\ell_1$ provided better results (Table \ref{table:table1_loss}), in accordance with our theoretical derivation. For MoDL we compared an image-domain $\ell_2$ loss, as proposed originally, and $\ell_1$ \emph{k}-space loss, and observed better performance for the latter as well (Table \ref{table:table1_loss}).  We thus use an $\ell_1$ \emph{k}-space loss for both \emph{k}-band and MoDL. SSDU was trained using the mixed $\ell_1$/$\ell_2$ loss suggested by its authors. Additionally, \emph{k}-square and \emph{k}-vertical were trained using an $\ell_1$ \emph{k}-space loss, for a fair comparison with \emph{k}-band.

\begin{table}[h]
\resizebox{\columnwidth}{!}{%
\begin{tabular}{|l|l|l|l|}
\hline
\textbf{Error Metric}   & \multicolumn{1}{c|}{SSIM} & \multicolumn{1}{c|}{NMSE} \\ \hline
\emph{K}-band (\textbf{\emph{k}-space $\ell_1$ loss}) & \textbf{0.959 ± 0.015}            & \textbf{0.00019 ± 0.00008}              \\ \hline
\emph{K}-band (\emph{k}-space $\ell_2$ loss) & 0.887 ± 0.037             & 0.00050 ± 0.00022            \\ \hline
MoDL (\textbf{\emph{k}-space $\ell_1$ loss})   & \textbf{0.956 ± 0.016}             & \textbf{0.00018 ± 0.00008}             \\ \hline
MoDL (image $\ell_2$ loss)     & 0.933 ± 0.023             & 0.00029 ± 0.00012             \\ \hline
\end{tabular}%
}
\caption{Comparison of loss functions for \emph{k}-band and MoDL with $R_{vd} = R_{band} = 4$.}
\label{table:table1_loss}
\end{table}

\section{Results}

\begin{figure}[h]
    \centering
    \includegraphics[width = 0.50\textwidth]{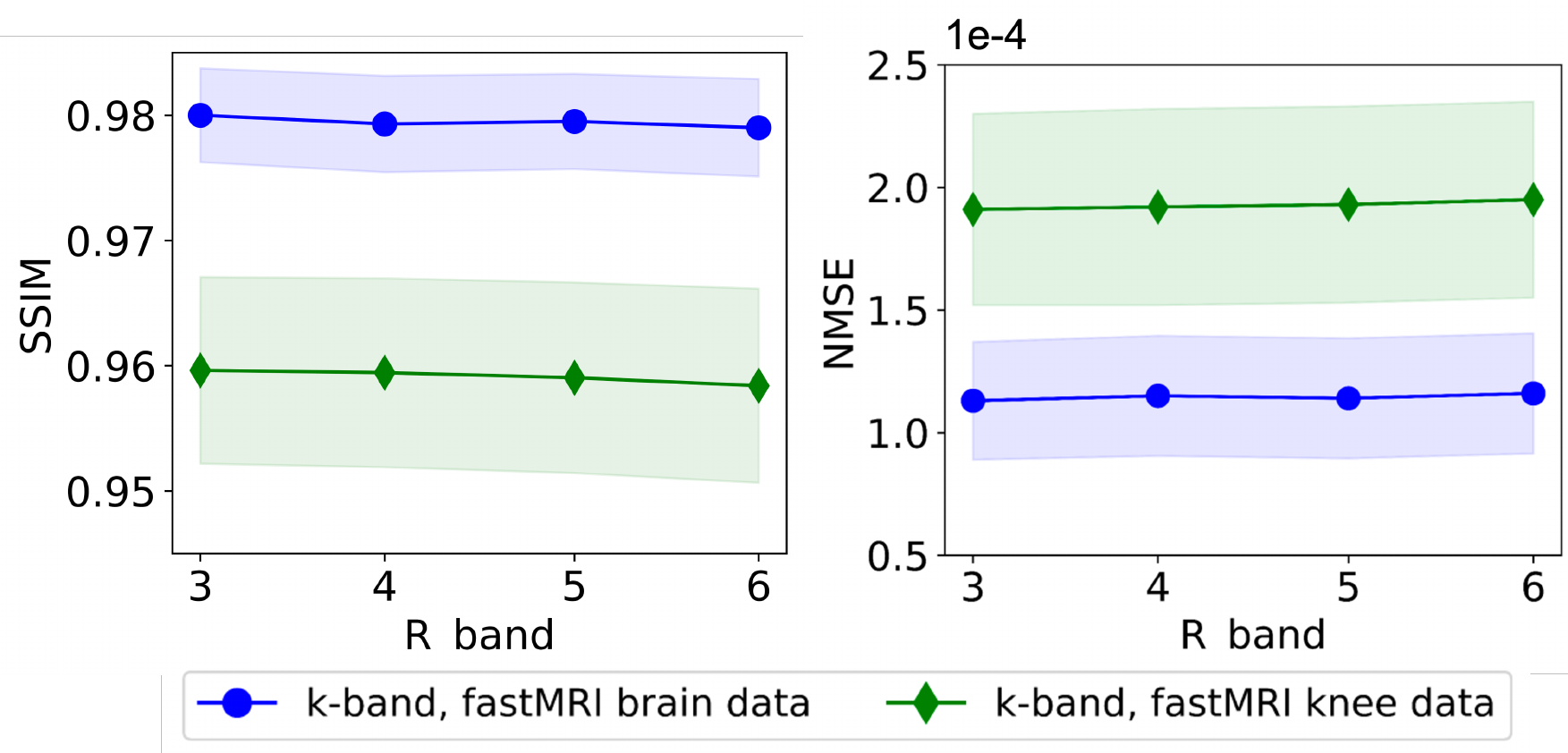}
    \caption{Evaluation of \emph{k}-band for a range of band widths ($R_{band}$) and fixed inference-time acceleration ($R_{vd}=4$) using the fastMRI knee (green) and brain (blue) datasets. Note that \emph{k}-band exhibits highly stable performance across a wide range of band widths for the training data (left-to-right).}
    \label{fig:fig_5_w_stats}
\end{figure}

\subsection{Stability over varying band sizes}

In the first experiment, we demonstrate the robustness of \emph{k}-band to smaller band widths. To do so, \emph{k}-band was trained using identical knee data but four different band widths. Here, $R_{band}$ was varied from 3 to 6; this corresponds to sampling approximately 33\% to 16\% of \emph{k}-space. To examine the reconstruction quality with respect to $R_{band}$ alone, the VD acceleration was held constant at $R_{vd} = 4$. This experiment employed 2D VD undersampling. The results (Figure \ref{fig:fig_5_w_stats}) are that
the SSIM and NMSE curves remained almost flat. This demonstrate that \emph{k}-band exhibits high performance stability over a range of band sizes.

\begin{figure}[h!]
    \centering
    \includegraphics[width = 0.35\textwidth]{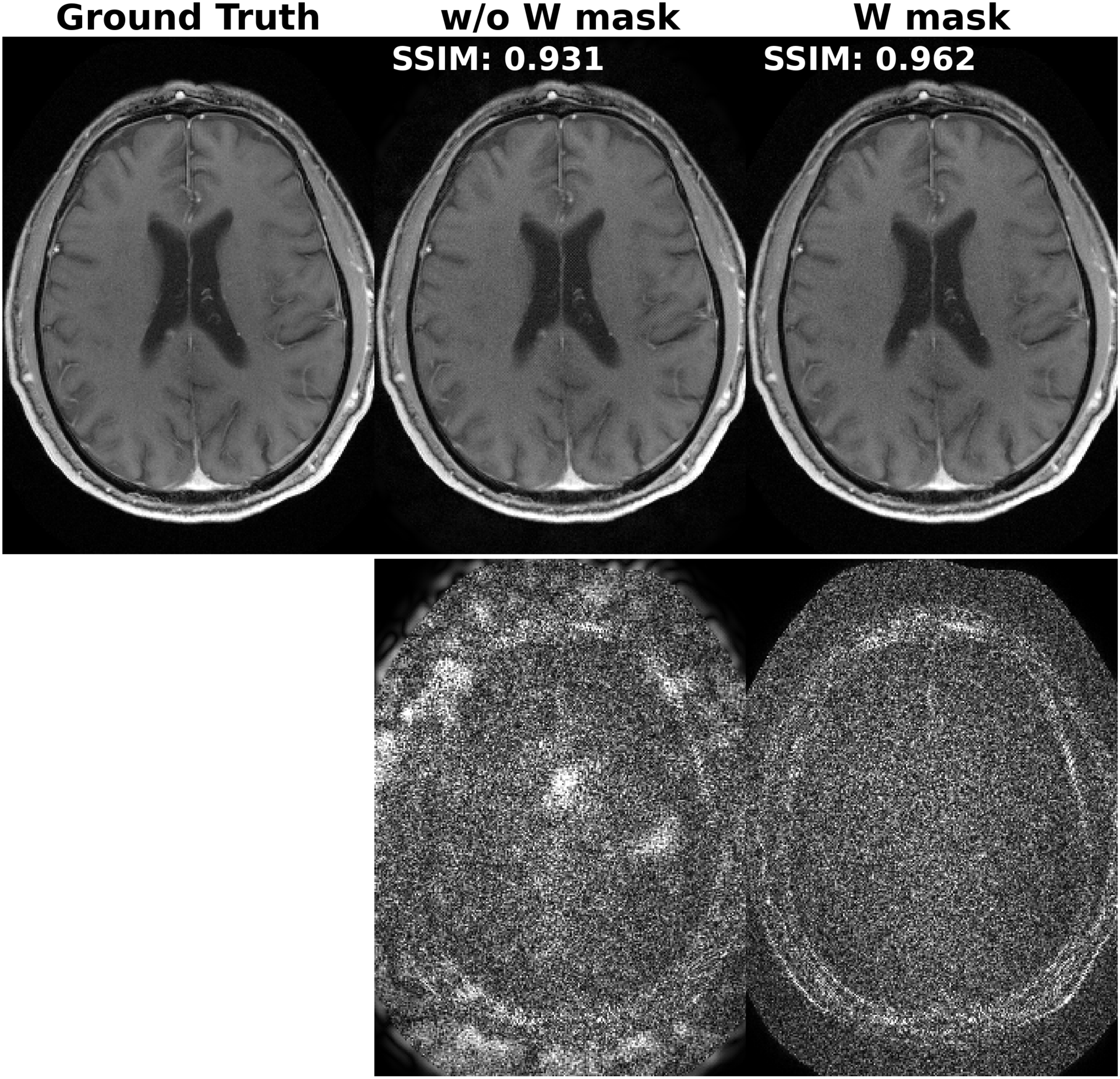}
    \caption{Comparison of training \emph{k}-band with and without the loss-weighting mask. Error maps, magnified by 10-fold, are shown below each reconstruction. Note that using the W mask leads to lower errors.}
    \label{fig:fig6_w_mask_compare}
\end{figure}

\subsection{Loss weighting}

To investigate the effect of the proposed $W$ mask, we trained \emph{k}-band with and without the mask. The experiments were conducted with $R_{band} = 4$ and $R_{vd} = 4$ using the brain data. The mean SSIM values were 0.964 ± 0.013 without the mask, and 0.980 ± 0.008 when the mask was used. A reconstruction example is shown in Figure \ref{fig:fig6_w_mask_compare}. These results suggest that training \emph{k}-band with the $W$ mask leads to overall sharper images.

\subsection{Comparison with other methods}

The five methods were compared by training them with undersampling factors of $R_{vd} = R_{band} = \{6, 8\}$. Experiments were done using the knee and brain data with 2D VD undersampling. The results shown in Figure \ref{fig:fig7_KneeRvdAll} (top, left column) indicate that \emph{k}-band performs substantially better than the two other methods trained on limited-resolution data (\emph{k}-vertical and \emph{k}-square), and comparably to the methods trained on high-resolution data (SSDU and MoDL). Furthermore, the reconstructed images displayed in Figure \ref{fig:fig10_knee_visual_result} (bottom) and \ref{fig:fig8_brain} suggest that \emph{k}-band obtains a visual quality similar to MoDL and SSDU, with the benefit of using only limited-resolution training data.

\begin{figure}[h]
    \centering
    \includegraphics[width = 0.49\textwidth]{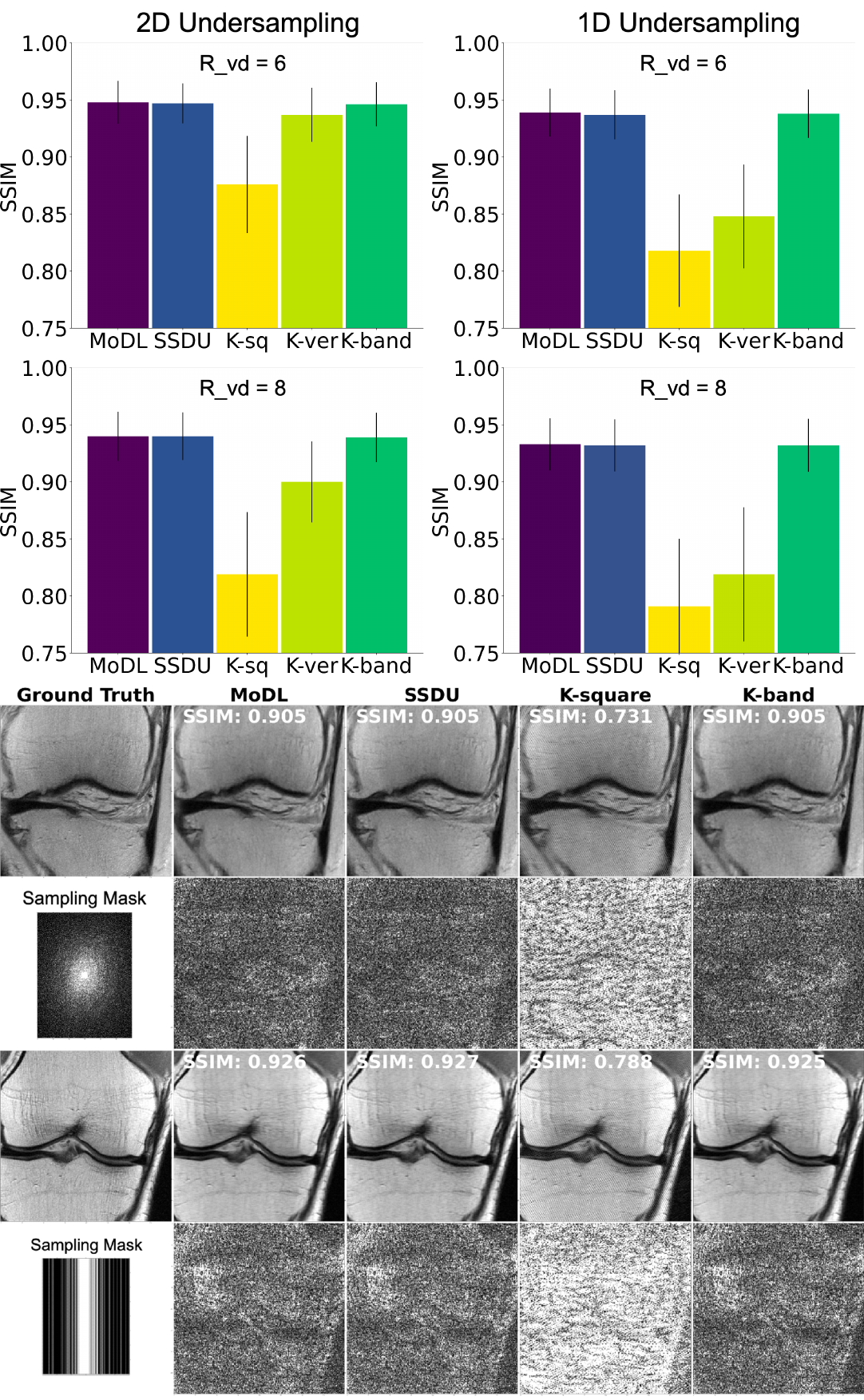}
    \caption{(Top) Comparison using the FastMRI knee data with both 2D and 1D VD undersampling during inference and $R_{vd} = R_{band} = \{6, 8\}$. Note that \emph{k}-band performs comparably to SSDU and MoDL, with the advantage of being trained using only limited-resolution data. (Bottom) Examples for knee images reconstructed using MoDL, SSDU, and \emph{k}-band, all trained with $R_{vd}=8$. For \emph{k}-band, the band width was $R_{band}=8$. Experiments were done with 2D and 1D undersampling.  Error maps, magnified by 10-fold, are shown below each reconstruction. Note that \emph{k}-band performs comparably to the other methods although it was trained using only limited-resolution data.}
    \label{fig:fig7_KneeRvdAll}\label{fig:fig10_knee_visual_result}
\end{figure}

\begin{figure}[h]
    \centering
    \includegraphics[width = 0.49\textwidth]{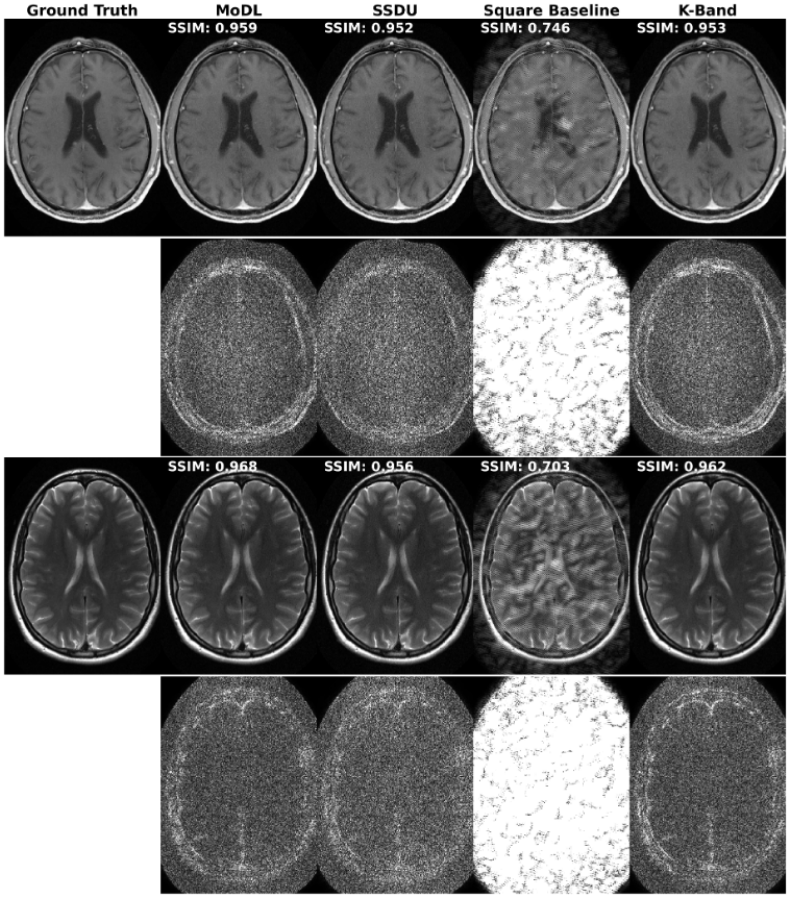}
    \caption{Reconstruction of the brain data with $R_{band} = R_{vd} = 6$. \emph{K}-band achieves better performance than \emph{k}-square and comparable performance to MoDL and SSDU, even though it was trained using only limited-resolution data. Error maps are magnified by 14-fold.}
    \label{fig:fig8_brain}
\end{figure}

We also trained the five methods using 1D VD undersampling masks. Note that 1D undersampling was not explored in the SSDU paper \cite{yaman2020self_mrm}. The statistical results presented in Figure \ref{fig:fig7_KneeRvdAll} (top, right column) indicate that \emph{k}-band performs comparably to MoDL and SSDU in this regime. The reconstructed images shown in Figure \ref{fig:fig10_knee_visual_result} (bottom) also indicate that \emph{k}-band obtains similar performance to MoDL and SSDU.

\subsection{Testing the number of angles $k$}
We investigate the performance of \emph{k}-band when it is trained on fewer band orientations, to test the stability of the framework and to explore a simpler implementation. We train separate networks and report results for the following numbers of band orientations $k$: $[5, 7, 10, 30]$. For each experiment, bands were sampled uniformly from corresponding orientation angles $[0, \frac{180}{k}, 2\times\frac{180}{k}, \cdots, (k-1)\times\frac{180}{k}]$ in order to sample as many distinct \emph{k}-space areas as possible. We also fixed $R_{vd} = 8$ and tested both $R_{band} = 4$ and $R_{band} = 8$ to see how the band area affects the number of bands necessary to match the performance of the vanilla $k=180$ orientations and the fully supervised MoDL.

\begin{figure}[h]
    \centering
    \includegraphics[width = 0.49\textwidth]{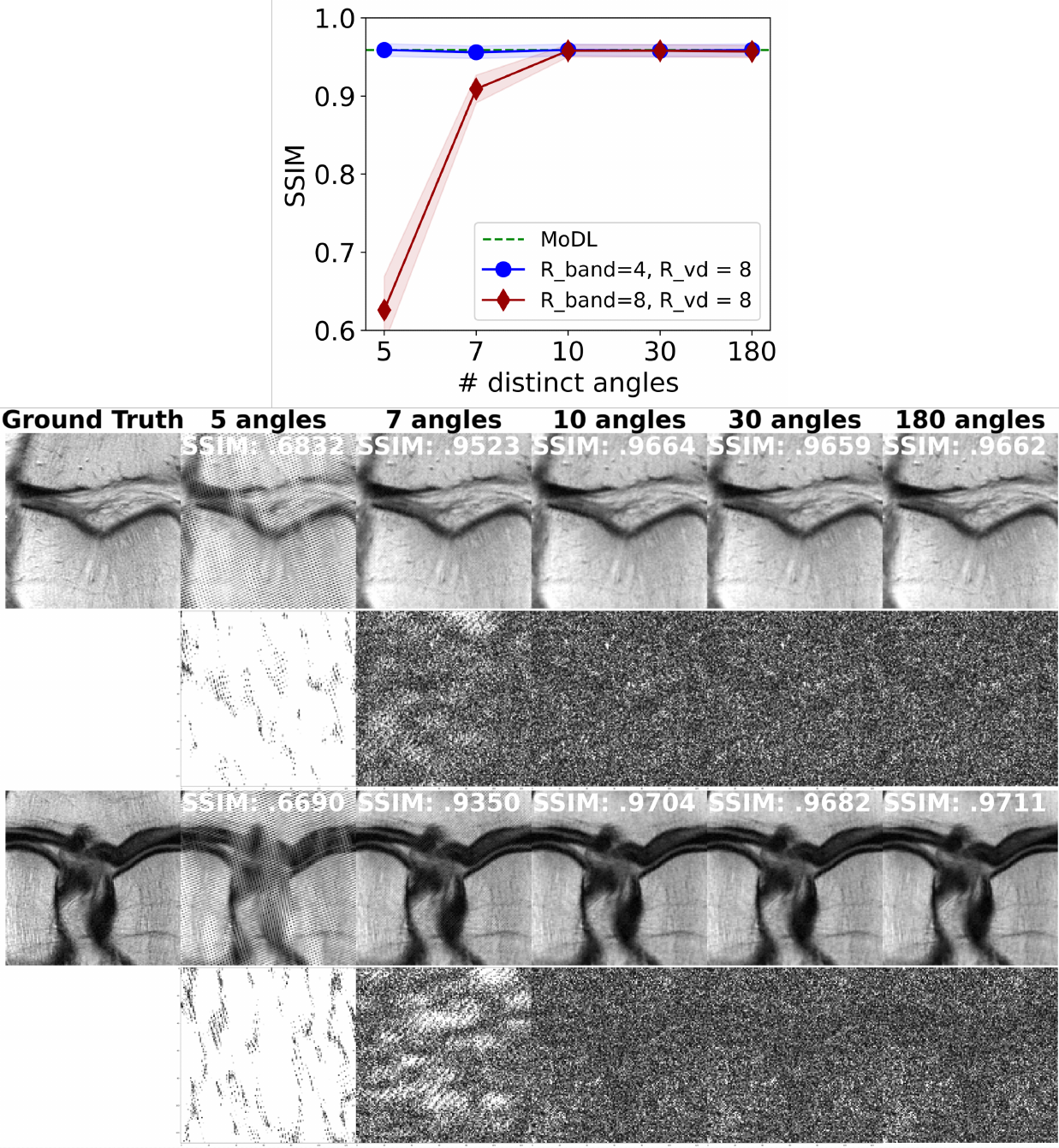}
    \caption{(Top) Comparison using the fastMRI knee data with $R_{vd} = 8$ and $R_{band} = [4,8]$ while varying the number of distinct band orientations for the training data. For $R_{band} = 4$, only 5 band orientations are necessary, while for $R_{band} = 8$ at least 10 band orientations are necessary. (Bottom) Visual comparison for the $R_{band} = 8$ case while varying the number of distinct band orientations $k = 5, 7, 10, 30, 180$. Error maps, magnified by 10-fold, are shown below each reconstruction. After $k = 10$ distinct orientations, we have stable performance.}
    \label{fig:fig10_band_limiting}
\end{figure}

Results are shown in Figure \ref{fig:fig10_band_limiting}. In the case of $R_{band} = 4$, we only need 5 distinct band orientations (namely at angles $[0, 36, 72, 108, 144]$) for high-quality, stable reconstruction, both visually and statistically. However, for the case of $R_{band} = 8$ we need at least 10 distinct band orientations for high-quality reconstruction. These results suggest that it is possible to obtain stable reconstruction when \emph{k}-band is trained on data acquired from fewer than 180 different orientations.

\section{Discussion}

Acquisition of fully sampled data is impractical in high-dimensional (e.g. dynamic/volumetric) MRI. This is a barrier to the development of DL reconstruction methods, which commonly require high-quality training data. To address these challenges, we introduce the \emph{k}-band framework, which involves co-designing the acquisition and training strategies. We propose the concept of \emph{SGD over k-space subsets}, where during training the gradients are computed using only \emph{k}-space bands, which can be acquired efficiently. We derive the method analytically and find two simple conditions under which \emph{k}-band stochastically approximates the desired (but impractical) SGD with fully sampled \emph{k}-space data: the randomization of the band orientation across scans and the use of a deterministic loss-weighting mask. We also demonstrate \emph{k}-band in an extensive set of experiments. Our results indicate that \emph{k}-band obtains performance on-par with fully sampled reconstruction methods and high  stability. Thus, despite being trained using only limited-resolution training data, \emph{k}-band is able to generalize to high-resolution data during inference.

The main merit in this study is the introduction of a framework that enables training networks using only datasets that were acquired with reduced resolution in one dimension. This approach can be particularly advantageous for dynamic MRI, where there is an inherent tradeoff between the spatial and temporal resolutions. Some applications include cardiac \cite{usman2013motion,schlemper2017deep}, dynamic contrast enhanced  \cite{zhang2015fast}, pulmonary \cite{jiang2018motion}, and 4D flow \cite{markl20124d}.  
Moreover, as the scarcity of open-access raw MRI databases is a current bottleneck hindering the development of DL reconstruction techniques, the proposed strategy can pave the way towards easier creation of new databases.

Another important advantage of \emph{k}-band is that its acquisition strategy can be implemented easily with numerous standard Cartesian MRI pulse sequences, which are the workhorses of clinical MRI. Furthermore, \emph{k}-band can be implemented without any special preparations, by simply limiting the resolution along the PE dimension and changing the scan orientation; this is possible on any standard MRI scanner. Moreover, \emph{k}-band is compatible with both 2D and 3D MRI acquisitions.

Our results demonstrate that \emph{k}-band offers high stability across a range of band widths (characterized by 1/$R_{band}$) (Figure \ref{fig:fig_5_w_stats}) and acceleration factors ($R_{vd}$) (Figure \ref{fig:fig7_KneeRvdAll}). Furthermore, the use of the proposed loss-weighting mask improves the performance (Figure \ref{fig:fig6_w_mask_compare}); the loss weighting penalizes reconstruction errors in high-frequency areas more heavily than in the \emph{k}-space center, compensating for the reduced exposure of the network to high-frequency areas during training.

We compared \emph{k}-band with four other methods, all implemented with the same network architecture. Our results indicate that \emph{k}-band performs 
substantially better than two baseline methods trained on limited-resolution data, and comparably to MoDL and SSDU, yielding highly similar results both visually and quantitatively (Figs. \ref{fig:fig7_KneeRvdAll} and  \ref{fig:fig8_brain}). \emph{k}-band hence reduces the need for acquiring high-resolution training data. 

The \emph{k}-band method shares some similarity with the SSDU and SelfCoLearn \cite{zou2022selfcolearn} methods, as all three methods use only undersampled data and self-supervised training. However, both SSDU and SelfCoLearn assume that the undersampling mask covers all of \emph{k}-space area, i.e. that the training data has a high resolution. In contrast, \emph{k}-band assumes that the training data are acquired only within a band, thus enabling efficient data acquisition.

Furthermore, the loss-weighting method introduced here exhibits similarities to the approaches proposed by Aggarwal et al \cite{aggarwal2021ensure} and Millard and Chiew\cite{millard2023theoretical}. The former suggests loss weighting based on Stein’s unbiased estimator for unsupervised training. However Stein’s unbiased estimator \cite{metzler2018unsupervised} makes stronger assumptions than losses based on Noise2Noise \cite{lehtinen2018noise2noise}, as the method considered in this paper. The latter (independently) proposed using \emph{k}-space loss-weighting in order to improve network robustness to the peripheral areas of \emph{k}-space that are infrequently sampled while inhibiting the center. However, their derivation of the loss-weighting used the Noisier2Noise\cite{moran2020noisier2noise} framework in order to analyze the expected value of SSDU reconstruction under a $\ell_2$ loss. Meanwhile, we analyze the moments of the gradients during training for any arbitrary stochastic sub-sampling strategy (extending to both bands sampling and SSDU variable-density sampling) and we also analytically show that an $\ell_1$ \emph{k}-space loss outperforms an $\ell_2$ \emph{k}-space loss for such stochastic acquisition.

Finally, \emph{k}-band is different from super-resolution \cite{van2012super, qiu2023medical}. Super-resolution methods pass in limited-resolution data as input during \emph{inference}, in order recover an image of higher resolution in the image domain. On the other hand, \emph{k}-band proposes to pass in limited-resolution \emph{k}-space data as input during \emph{training}, to enable easier acquisition of training data. Furthermore, during \emph{inference}, \emph{k}-band is not limited to reconstructing limited-resolution data, unlike super-resolution methods.

This work also has some limitations. The proposed method was validated using retrospective simulations with the fastMRI database \cite{knoll2020fastmri}. Furthermore, although our code is suitable for multi-coil data, our networks were trained on coil-combined data due to computational constraints. Future work will include validation of the method on prospectively acquired multi-coil data. 

\section{Conclusion}

We introduce \emph{k}-band, an acquisition-reconstruction framework that enables training self-supervised DL models using only limited-resolution data, with test-time generalization to high-resolution data. This work introduces the framework of \emph{SGD over \emph{k}-space subsets}, which enables rapid and easy data acquisition using any commercial MRI scanner, and offers a practical solution for training DL models in data-challenging regimes. Numerical experiments demonstrate that \emph{k}-band reconstructs images with quality competitive with those of supervised and self-supervised techniques trained on high-resolution data, thereby reducing the need for acquisition of such data. This work can hence pave the way towards development of new DL techniques for high-dimensional MRI. 

\section*{Acknowledgments}
The authors would like to thank the National Institute of Biomedical Imaging and Bioengineering (grant numbers U24EB029240, U01EB029427, R01EB009690), GE HealthCare, and the Weizmann Institute Women’s Postdoctoral Career Development Award in Science for financial support.

\section*{Data Availability Statement}
Our code is available at 
\url{https://github.com/mikgroup/K-band}.

\bibliography{MRM-AMA}%
\end{document}